\documentclass[11pt]{amsart}

\usepackage[english]{babel}
\usepackage[latin1]{inputenc}
\usepackage[a4paper,top=3cm,bottom=3cm,left=2.5cm,right=2.5cm, bindingoffset=5mm]{geometry}
\usepackage{mathpple}
\usepackage{newlfont}
\usepackage{amsmath}
\usepackage{graphicx}
\usepackage{amssymb}
\usepackage{amsmath}
\usepackage{latexsym}
\usepackage{amsthm}
\usepackage{stmaryrd}
\usepackage{braket}
\usepackage{enumerate}
\usepackage{booktabs}
\usepackage{mathtools}
\usepackage{array}
\usepackage{paralist}
\usepackage{cleveref}
\usepackage{color}

\usepackage{pxfonts}

\newtheorem{Theorem}{Theorem}%[section]

\theoremstyle{definition}
\newtheorem{Def}{Definition}
\newtheorem{Rem}[Theorem]{Remark}

\newtheorem*{problemg}{Generalized MinRank Problem}
\newtheorem*{problemc}{Classical MinRank Problem}

\DeclareMathOperator{\rank}{rank}

\DeclareMathOperator{\solvdeg}{solv.deg}

\DeclareMathOperator{\reg}{reg}

\DeclareMathOperator{\codim}{codim}

\newcommand{\kk}{\Bbbk}            %%%%%% Notazione per il campo residuo
\title{The complexity of MinRank} %{The Min Rank Problem revisited}
\author{Alessio Caminata and Elisa Gorla}

\address{{\small Alessio Caminata, Dipartimento di Matematica, Universit\`a di Genova\\ via Dodecaneso 35, 16146, Genova, Italy}}
\email{{\small caminata@dima.unige.it}}

\address{{\small Elisa Gorla, Institut de Math\'{e}matiques, Universit\'{e} de Neuch\^{a}tel\\Rue Emile-Argand 11, CH-2000 Neuch\^{a}tel, Switzerland}}
\email{{\small elisa.gorla@unine.ch}}

\subjclass[2010]{Primary: 94A60, 13P10, 13P15, 13C40, 13P25.}

\keywords{MinRank Problem; minors; solving degree; Castelnuovo-Mumford regularity; Gr\"obner bases; multivariate cryptography; post-quantum cryptography.}

\begin{document}

\begin{abstract}
In this note, we leverage the results of~\cite{CG19} to produce a concise and rigorous proof for the complexity of the generalized MinRank Problem in the under-defined and well-defined case. Our main theorem recovers and extends the main results of \cite{FSS10,FSS13}.
\end{abstract}
	
\maketitle

%%%%%%%%%%%%%%%%%%%%%%%%%%%%%%%%%%%
%%%%%%% SECTION %%%%%%%%%%%%%%%%%%%%%%%
\section{Introduction}\label{section-introduction}

The MinRank Problem asks to find an element of least rank in a given space of matrices. In its classical formulation, one searches for a matrix of minimum rank in a vector space, given via a system of generators.

\begin{problemc}
Let $\kk$ be a field and let $m,n,r,k$ be positive integers. Given as input $k$ matrices $M_1,\dots,M_k$ with entries in $\kk$, find $x_1,\dots,x_k\in\kk$ such that the corresponding linear combination satisfies
\[
\rank\left(\sum_{i=1}^kx_iM_i\right)\leq r.
\] 	
\end{problemc}

The entries of the matrix $M=\sum_{i=1}^kx_iM_i$ are linear polynomials in the variables $x_1,\dots,x_k$. The following is a natural generalization of the MinRank Problem.

\begin{problemg}
Let $\kk$ be a field and let $m,n,r,k$ be positive integers. Given as input a matrix $M$ with entries in $\kk[x_1,\dots,x_k]$, compute the set of points in $\kk^k$ at which the evaluation of $M$ has rank at most $r$.
\end{problemg}

Both of these problems arise naturally within cryptography and coding theory, as well as in numerous other applications. Within multivariate cryptography, the MinRank Problem plays a central role in the cryptanalysis of several systems, including HFE and its variants~\cite{KS99, BFP13, CSV17, VS17, DPPS18}, the TTM Cryptosystem~\cite{GC00}, and the ABC Cryptosystem~\cite{MSP14, MPS17}. Within coding theory, the problem of decoding a linear rank-metric code is always an instance of the MinRank Problem, and in some cases it can be modeled as a generalized MinRank Problem, where some entries of the matrix have degree greater than one, see e.g.~\cite{MGR08, GMR12}. Further applications of the generalized MinRank Problem to nonlinear computational geometry, real geometry and optimization, and other problems in symbolic computation are discussed in the introduction of ~\cite{FSS13}.

Following~\cite{KS99}, we distinguish the following three situations.

\begin{Def}
A MinRank Problem is \emph{under-defined} if $k>(n-r)(m-r)$, \emph{well-defined} if $k=(n-r)(m-r)$, and \emph{over-determined} if $k<(n-r)(m-r)$.	
\end{Def} 

\par There are at least three ways of approaching the MinRank Problem: the Kipnis-Shamir modeling~\cite{KS99}, the linear algebra search~\cite{GC00}, and the minors modeling. We concentrate on the latter. 
The minors modeling relies on the following observation:
A vector $(a_1,\dots,a_k)$ is a solution of the (classical or generalized) MinRank Problem for a matrix $M$ if and only if all minors of size $r+1$ of $M$ vanish at this point. Thus we can find the solutions of the generalized MinRank Problem by solving the polynomial system consisting of all minors of size $r+1$ of $M$.
This is a system of multivariate polynomial equations $\mathcal{F}=\{f_1,\dots,f_s \}$, so one may attempt to solve it by means of the usual Gr\"obner bases methods.
The complexity of these methods is controlled by the \emph{solving degree} of $\mathcal{F}$, that is the highest degree of polynomials appearing during the computation of a degree reverse lexicographic Gr\"obner basis of $\mathcal{F}$.

In this paper, we take another look at the complexity of solving the generalized MinRank Problem with the minors modeling. We focus on the under-defined and well-defined situations, which we treat with a unified approach. Notice that no fully provable, general results on the complexity of the over-defined case are currently available.

The results from~\cite{CG19}, in combination with classical commutative algebra results, provide us with a simple provable estimate for the complexity of the homogeneous version of the generalized MinRank Problem. %More generally, Theorem~\ref{thm:main} holds in the situation when the minors of the matrix obtained by homogenizing the entries of $M$ are the homogenization of the minors of $M$. 
As a special case of our main result, we obtain a simple and concise proof of the main results from~\cite{FSS10,FSS13}, which avoids lengthy technical computations.

%%%%%%%%%%%%%%%%%%%%%%%%%%%%%%%%%%%
%%%%%%% SECTION %%%%%%%%%%%%%%%%%%%%%%%
\section{Main Results}\label{section-preliminaries}
We fix an infinite field $\kk$ and positive integers $m,n,r,k$. Without loss of generality, we assume that $n\geq m$ and $r<m$. We focus on the MinRank Problem in the under-defined and well-defined case. We state the results in increasing order of generality.

\begin{Theorem}[{{\cite[Corollary~4]{FSS10}}}]\label{thm:square}
	The solving degree of the minors modeling of a generic classical well-defined square MinRank Problem ($m=n$ and $k=(n-r)^2$) is upper bounded by
	\[
	\solvdeg(\mathcal{F})\leq nr-r^2+1. 
	\]  
\end{Theorem}

%\begin{Theorem}[\cite{FSS13}, Corollary~19]\label{thm:linear}
%	Let $M$ be an $m\times n$ matrix whose entries are generic linear polynomials in $\kk[x_1,\dots,x_k]$ and assume $k\geq(m-r)(n-r)$. Let $\mathcal{F}$ be the polynomial system of the minors of size $r+1$ of $M$. 
%	Then the solving degree of $\mathcal{F}$ is upper bounded by
%	\[
%	\solvdeg(\mathcal{F})\leq mr-r^2+1.
%	\] 
%\end{Theorem}

\begin{Theorem}[{{\cite[Lemma~18, Corollary~19, Lemma~22, Corollary~23]{FSS13}}}]\label{thm:degd}
	Let $M$ be an $m\times n$ matrix whose entries are generic homogeneous polynomials of degree $d$ in $\kk[x_1,\dots,x_k]$ and assume $k\geq(m-r)(n-r)$. Let $\mathcal{F}$ be the polynomial system of the minors of size $r+1$ of $M$. 
	Then the solving degree of $\mathcal{F}$ is upper bounded by
	\[
	\solvdeg(\mathcal{F})\leq (m-r)(nd-n+r)+1.
	\] 
\end{Theorem}

\par The previous theorems recover the main results of \cite{FSS10,FSS13}. We obtain them as a consequence of our  more general Theorem~\ref{thm:main}, by letting $m=n$ and $d_{i,j}=1$ (Theorem~\ref{thm:square}), or $d_{i,j}=d$ (Theorem~\ref{thm:degd}).

We consider an $m\times n$ matrix  $M$, whose entry in position $(i,j)$ is a polynomial of degree $d_{i,j}$ in $\kk[x_1,\dots,x_k]$, for all $i,j$. Up to permuting the rows of $M$, we may assume that $d_{1,1}\leq d_{2,1}\leq\cdots\leq d_{m,1}$. Moreover, assume that the following two conditions hold:
\begin{compactenum}
	\item $d_{i,j}>0$ for all $i,j$.
	\item $d_{i,j}+d_{h,\ell}=d_{i,\ell}+d_{h,j}$ for all $i,j,\ell,h$.
\end{compactenum}	
Finally, we assume that the entries of $M$ are generic polynomials. One may think of this assumption as the coefficients of each polynomial being randomly chosen.

\begin{Theorem}\label{thm:main}
	Let $M$ be an $m\times n$ matrix as above and assume $k\geq(m-r)(n-r)$. Let $\mathcal{F}$ be the polynomial system of the minors of size $r+1$ of $M$. 
	Then the solving degree of $\mathcal{F}$ is upper bounded by
	\[
	\solvdeg(\mathcal{F})\leq (m-r)\sum_{i=1}^r d_{i,i}+\sum_{i=r+1}^m\sum_{j=r+1}^n d_{i,j}-(m-r)(n-r)+1.
	\] 
\end{Theorem}

\begin{proof}
	Under our assumptions, the homogenizations of the $(r+1)$-minors of $M$ are the $(r+1)$-minors of the matrix obtained from $M$ by homogenizing its entries. Therefore, we may assume without loss of generality that the entries of $M$ are generic homogeneous polynomials.
The main result of \cite[Section~3.3]{CG19} implies that 
\[
\solvdeg(\mathcal{F})\leq\reg I,
\]
where $I$ is the ideal generated by the polynomials of $\mathcal{F}$ and $\reg I$ denotes the Castelnuovo-Mumford regularity of $I$.
We can compute it as follows. 
\par First, since the polynomials of $M$ are generic and the matrix $M$ is homogeneous, by combining Eagon-Northcott's Theorem~\cite[Theorem~3]{EN62} with \cite[Theorem~2.5]{BV88} one obtains that the quotient ring $S=\kk[x_1,\dots,x_k]/I$ is Cohen-Macaulay and the ideal $I$ has codimension $\codim(I)=(m-r)(n-r)$. Recall that the codimension of a homogeneous ideal in a polynomial ring $\kk[x_1,\dots,x_k]$ is the difference between $k$ and the Krull dimension of the quotient of the polynomial ring by the ideal.

Now consider the quotient ring $T=\kk[X]/I_{r+1}(X)$, where $X=(x_{i,j})$ is a matrix of size $m\times n$ whose entries are distinct variables, $\deg(x_{i,j})=d_{i,j}$, $\kk[X]$ is the polynomial ring over $\kk$ with variables the entries of $X$, and $I_{r+1}(X)$ denotes the ideal generated by the minors of size $r+1$ of $X$. By~\cite[Corollary~4]{HE71} $\codim(I_{r+1}(X))=(m-r)(n-r)$, see also~\cite[Theorem~3.7.1]{BH98}.

Since $\codim(I)=\codim(I_{r+1}(X))$, by~\cite[Theorem~3.5]{BV88} a minimal graded free resolution of $S$ is obtained from a minimal graded free resolution of $T$ by substituting $x_{i,j}$ with the entry of $M$ in position $(i,j)$, for all $i$ and $j$. In particular
\[
\reg_{\kk[x_1,\ldots,x_k]}(S)=\reg_{\kk[X]}(T),
\]
where $\reg(S)=\reg(I)-1$ and $\reg(T)=\reg (I_{r+1}(X))-1$.
Moreover, since $T$ is Cohen-Macaulay, we can express its regularity in terms of its $a$-invariant  (see \cite[Definition~3.6.13]{BH98}) and of the codimension of $I_{r+1}(X)$. 
We have $$\reg(T)=a(T)-a(\kk[X])-\codim(I_{r+1}(X))=a(T)+\sum_{i=1}^m\sum_{j=1}^n d_{i,j}-(m-r)(n-r),$$ where $a$ denotes the $a$-invariant, the first equality follows from~\cite[Examples 3.6.15 b)]{BH98}, and the second from~\cite[Examples 3.6.15 a)]{BH98} and $\codim(I_{r+1}(X))=(m-r)(n-r)$.
By~\cite[Corollary 1.5]{BH92} 
\[
a(T)=-r\sum_{i=1}^md_{i,i}-\sum_{i=1}^r\sum_{j=m+1}^nd_{i,j},
\]
where $d_{i,j}=e_i+f_j$ in the notation of \cite{BH92}.
Putting everything together we obtain
\[\begin{aligned}
	\reg(I)&=\reg(S)+1=a(T)+\sum_{i=1}^m\sum_{j=1}^n d_{i,j}-(m-r)(n-r)+1
	\\&=(m-r)\sum_{i=1}^r d_{i,i}+\sum_{i=r+1}^m\sum_{j=r+1}^n d_{i,j}-(m-r)(n-r)+1,
\end{aligned}\]
which proves the statement.
\end{proof}

\begin{Rem}
	Theorem~\ref{thm:main} analyzes the under-defined and well-defined situations.
	In the over-defined situation, assume that $k$ is sufficiently small and that $d_{i,j}=1$ for all $i$ and $j$. Then the minors of size $r+1$ of $M$ generate the maximal ideal to the power $r+1$. In particular, $$\solvdeg(\mathcal{F})=r+1.$$
\end{Rem}

\begin{Rem}\label{finitefield}
	The word ``generic'' used in the statements is a technical term from algebraic geometry, which means ``there exists a nonempty open set'' of polynomials for which the result holds. This is exactly the same use of generic as in~\cite{FSS10,FSS13}. We stress that the genericity assumption is often essential to a type of approach that uses algebraic geometry. To the extent of our knowledge, this assumption appears also in all the previous works that use similar methods. 
	
	Usually one thinks of a generic property as a property that holds for ``almost every point'' of the ambient space. In order for this intuition to be true, however, one needs to work over an infinite field, or at least over a large enough field extension of $\kk$ (if $\kk$ is a finite field). In fact, a nonempty open set over an infinite field may contain only a few points, or even no point, over a given finite subfield. 
	
	One may therefore be lead to think that theorems with a genericity assumption are of little use over finite fields. This is however not the case. In fact, if an open set is nonempty over the algebraic closure, then it will contain most points over a large enough (but finite) field extension of $\kk$. Therefore, if we are willing to take a field extension, we have that a generic property holds for most points.
	
	In addition, any open set  
	%is the complement of the zero locus of a finite number of polynomials. In particular, it 
	is defined by a finite number of conditions. Whenever one can explicitly describe them, 
	%e.g. in the form of equations that should not be satisfied, 
	one can check whether any given point (including points over any finite field) satisfies them, which is equivalent to checking whether the point belongs to the open set. These conditions may always be expressed as a set of polynomial equations which should not all vanish on the point in question. Sometimes, when the polynomials are difficult to describe explicitly or involve a large number of terms, one may choose to describe the conditions as equivalent properties that can be checked directly. E.g., in the proof of Theorem~\ref{thm:main}, for any minor of the matrix $M$ one can check whether the homogenization of the minor is equal to the corresponding minor of the matrix obtained from $M$ by homogenizing its entries. This condition can be expressed also as a polynomial in the coefficients of the entries of $M$, namely the condition on the homogenization holds if and only if the polynomial does not vanish on the coefficients of the entries of $M$. In particular, whenever we are able to explicitly state the genericity conditions, one can directly check whether a given system of equations satisfies the genericity properties, independently of the field of definition (which can also have small cardinality).
	%This corresponds to the fact that the coefficients of the entries of $M$ involved in the minor do not satisfy the polynomial equation corresponding to the vanishing of the coefficients of the top degree term.
	\phantom\qedhere
\end{Rem}

In the next theorem we explicitly state the genericity conditions of Theorem~\ref{thm:main}, so that they can be checked directly over any finite field. This provides a version of Theorem~\ref{thm:main} over finite fields.

\begin{Theorem}
	Let $\kk$ be a finite field.
	Let $M$ be an $m\times n$ matrix whose entry in position $(i,j)$ is a polynomial of degree $d_{i,j}>0$ in $\kk[x_1,\dots,x_k]$, for all $i,j$. Assume that $k\geq(m-r)(n-r)$, $d_{1,1}\leq d_{2,1}\leq\cdots\leq d_{m,1}$, and $d_{i,j}+d_{h,\ell}=d_{i,\ell}+d_{h,j}$ for all $i,j,\ell,h$.
	Let $\mathcal{F}$ be the polynomial system of the minors of size $r+1$ of $M$. %and let $I$ be the ideal generated by $\mathcal{F}$.
	Let $t$ be a new variable, let $M^h$ be the matrix obtained from $M$ by homogenizing its entries with respect to $t$, and let $J=I_{r+1}(M^h)$.
	Suppose that $\codim(J)=(m-r)(n-r)$, that $t\nmid 0$ modulo $J$, and that the homogenization with respect to $t$ of each $(r+1)$-minor of $M$ equals the corresponding $(r+1)$-minor of $M^h$.
	Then the solving degree of $\mathcal{F}$ is upper bounded by
	\[
	\solvdeg(\mathcal{F})\leq (m-r)\sum_{i=1}^r d_{i,i}+\sum_{i=r+1}^m\sum_{j=r+1}^n d_{i,j}-(m-r)(n-r)+1.
	\] 
\end{Theorem}

\section*{Acknowledgements}
	 We are grateful to an anonymous referee for a detailed reading and comments which helped us improve the clarity of the proof of the main theorem.

\label{lastpage}
\end{document}